\documentclass[runningheads]{llncs}

\usepackage{rotating} 

\newif\ifEditMode
\EditModefalse

\usepackage[T1]{fontenc}        
\usepackage{amsmath, amssymb, amsfonts, amscd, amsbsy}
\usepackage{mathtools}
\usepackage{graphicx, tikz, tikz-cd}
\usepackage{enumerate, epstopdf, cite,verbatim}
\usepackage{wrapfig}
\usepackage[inline]{enumitem}
\usepackage{url}
\usepackage{xcolor}
\usepackage{booktabs, subcaption}
\usepackage{scalerel, xspace}
\usepackage[normalem]{ulem}
\usepackage{prftree}
\usepackage{hyperref}
\usepackage{algorithm}
\usepackage{algpseudocode}
\usepackage[font=itshape]{quoting}
\usepackage{backnaur}
\usetikzlibrary{positioning}
\usetikzlibrary{shapes.geometric}

\usepackage{amsthm}
\usepackage[breakable, theorems, skins]{tcolorbox}

\PassOptionsToPackage{bookmarks=false}{hyperref}
\tikzcdset{
arrow style=tikz,
diagrams={>={Straight Barb[scale=0.8]}}
}

\allowdisplaybreaks

\ifEditMode
   \paperwidth=\dimexpr \paperwidth + 8cm\relax
   \oddsidemargin=\dimexpr\oddsidemargin + 4cm\relax
   \evensidemargin=\dimexpr\evensidemargin + 4cm\relax
   \marginparwidth=\dimexpr \marginparwidth + 4cm\relax
\fi

\DeclareCaptionFormat{algor}{%
  \hrulefill\par\offinterlineskip\vskip1pt%
    \textbf{#1#2}#3\offinterlineskip\hrulefill}
\DeclareCaptionStyle{algori}{singlelinecheck=off,format=algor,labelsep=space}

\providecommand{\setunion}[0]{\boldsymbol{\cup}}
\providecommand{\setint}[0]{\boldsymbol{\cap}}
\providecommand{\meet}[0]{\wedge}

\providecommand{\join}[0]{\vee}

\providecommand{\setArg}[2]{\left\{ {#1} \,\,\middle|\,\, {#2} \right\}}

\providecommand{\cmp}[0]{M}

\providecommand{\behaviorUnv}[0]{\mathcal{B}}
\providecommand{\cmpComposition}[0]{\parallel}

\providecommand{\cont}{\mathcal{C}}

\renewcommand{\paragraph}[1]{\noindent\textbf{#1}}

\definecolor{navy}{rgb}{0.0, 0.0, 0.5}
\newcommand{\Question}[1]{{\color{blue}#1}}

\providecommand{\behaviorset}[0]{\mathcal{B}}
\providecommand{\sat}[0]{\models}
\providecommand{\envsat}[0]{\models^E}
\providecommand{\impsat}[0]{\models^M}

\begin{document}


\title{Some Algebraic Aspects of \\ Assume-Guarantee Reasoning\thanks{This paper is based on Chapter 6 of \cite{Incer:EECS-2022-99}.}}

\author{Inigo Incer\inst{1}
\and
Albert Benveniste\inst{2}
\and
Alberto Sangiovanni-Vincentelli\inst{3}
}
\authorrunning{Incer, Benveniste, and Sangiovanni-Vincentelli}

\institute{California Institute of Technology, Pasadena, CA, USA
\and
INRIA/IRISA, Rennes, France
\and
University of California, Berkeley, CA, USA
}

\maketitle

\begin{abstract}
We present the algebra of assume-guarantee (AG) contracts. We define contracts, provide new as well as known operations, and show how these operations are related. Contracts are functorial: any Boolean algebra has an associated contract algebra. We study monoid and semiring structures in contract algebra---and the mappings between such structures. We discuss the actions of a Boolean algebra on its contract algebra.
\end{abstract}

\providecommand{\id}[0]{e}
\newcommand{\algSummPath}[0]{sections/}
\newcommand{\documentPath}{\algSummPath}

\newcommand{\calg}[0]{\mathbf{C}}
\newcommand{\andmon}{\calg_\land^M}
\newcommand{\ormon}{\calg_\lor^M}
\newcommand{\parmon}{\calg_\parallel^M}
\newcommand{\mermon}{\calg_\bullet^M}

\newcommand{\andsr}{\calg_\land^S}
\newcommand{\orsr}{\calg_\lor^S}
\newcommand{\parsr}{\calg_\parallel^S}
\newcommand{\mersr}{\calg_\bullet^S}

\newcommand{\balg}[0]{B}

\newcommand{\boolandmon}[0]{\mathbf{M}_{\land}}
\newcommand{\boolormon}[0]{\mathbf{M}_{\lor}}
\newcommand{\boolandsr}[0]{\mathbf{S}_{\land}}
\newcommand{\boolorsr}[0]{\mathbf{S}_{\lor}}


\section{Introduction}

The design of complex cyber-physical systems (CPS) involves fundamental challenges in modeling, specification, and integration.
Among the modeling challenges, the need to model the interconnection of components modeled using discrete transitions with those using differential equations is fundamental. It is the opposition between the discrete and the continuous, which Ren{\' e} Thom calls ``the fundamental aporia of mathematics.'' Specification is well known to be hard in system engineering: practitioners consistently rank the generation of specifications for a project among the top challenges in system design \cite{national2016top,sage1998systems}. Finally, integration has seen steep increases in magnitude in the 20th century, bringing technical challenges to engineering (how do we design, build, and maintain such systems?)\footnote{We read, for example:
``We are increasingly experiencing a new type of accident that
arises in the interactions among components (electromechanical, digital, and human)
rather than in the failure of individual components'' \cite{LEVESON2004237}. 
``Almost all SI [system integration] failures occur at interfaces primarily due to incomplete, inconsistent, or misunderstood specifications'' \cite{https://doi.org/10.1002/sys.21249}.
}
and organizational challenges to the business landscape (how does an organization change to support the development, construction, and maintenance of such systems?)\footnote{Hobday et al.~\cite{hobday2005systems} argue that ``systems integration has evolved beyond its original technical and operational tasks to encompass a strategic business dimension becoming, therefore, a core capability of many high-technology corporations\textellipsis The more complex, high-technology, and high cost the product, the more significant systems integration becomes to the productive activity of the firm\textellipsis Systems integration capabilities are inextricably linked to decisions on whether to make in-house, outsource, or collaborate in production and competition.''
Davies et al.~\cite{davies2007organizing} claim that
``
the traditional advantages of the vertically-integrated systems
seller offering single-vendor designed systems is no longer a
major source of competitive advantage in many industries.''}.
The organizational challenges of system engineering are so salient that some authors categorize the field as a branch of management \cite{https://doi.org/10.1002/sys.20028,https://doi.org/10.1002/inst.201013143}.

Werner Damm has championed two key concepts to address system design challenges: rich components and contracts.
In \cite{DammRichComponents}, he observes that ``the design of components in complex systems
inherently involves multi-site, multi-domain and cross-organizational design teams'' and that
``in spite of well-defined and enforced process models, such as the V-model, a multitude of `disturbances' may lead to
undesirable design iterations. `Disturbances' in this
context take a variety of forms, such as late or
incomplete sets of requirements, late requirement
changes, unspecified assumptions, unexpected
disruption in a supply chain sub-process, the failure to
meet non-functional constraints such as
communication latencies, implicit interdependencies,
etc.''
One way to deal with these disturbances is by characterizing the domains of validity of our models.
In \cite{DammRichComponents}, a rich component is proposed as a classical component upgraded by
\begin{quote}
    (1) extending component specifications to cover
    all viewpoints necessary for electronic
    system design;
    (2) explicating the dependency of such
    specifications on assumptions on the context
    of a component;
    (3) providing classifiers to such assumptions,
    relating both the positioning in a layered
    design space (horizontal, up, down), as well
    as to their confidence level.
\end{quote}

The idea of representing components using assume-guarantee (AG) contracts launched a major research effort in system engineering \cite{multViewpoint,sangiovanni2012taming,iannopollo2014library,7268792,BenvenisteContractBook}. We wanted to streamline two key aspects of complex system design: relationships between a system integrator and its suppliers (cross-organizational design) and the interactions between multiple engineering organizations within the same company (cross-domain design).
In contract-based design, to each component in our system we assign an assume-guarantee specification---or contract.
A rich contract algebra \cite{BenvenisteContractBook,Incer:EECS-2022-99} allows us to relate global system properties to the local properties of the components it comprises. This algebra provides mathematical support for key aspects of the system design process. In this paper, we present this algebra and discuss its role in system design.

\medskip

\paragraph{Contracts.}
AG contracts can be understood as formal specifications split in two parts: (i) assumptions made on the environment, and (ii) responsibilities assigned to the object satisfying the specification when it is instantiated in an environment which meets the assumptions of the contract. 
Contracts were introduced to streamline the integration of complex systems and to support concurrent design. System integration pertains to the composition of multiple design objects into a coherent whole. For example, suppose a company wishes to implement a system with a specification $\cont$; designers may realize that there are two sub-specifications $\cont_1$ and $\cont_2$ such that the composition of their implementations always yields an implementation for the top level specification. In the language of AG contracts, we would say that the composition of $\cont_1$ and $\cont_2$, written $\cont_1 \parallel \cont_2$, \emph{refines} $\cont$. This company may now develop an implementation $\cmp_1$ for $\cont_1$, and assign $\cont_2$ to a third-party OEM to deliver an implementation $\cmp_2$.
If $\cmp_2$ is an implementation for $\cont_2$, the original company knows that $\cmp_1$ and $\cmp_2$ can be composed and that this composition meets the top-level specification $\cont$.
In this setting, frictions in the supply chain are alleviated as companies exchange formal specifications expressed as contracts. 

An additional use of contracts is as follows. Suppose our company wants to implement a system with a specification $\cont$  using a component $\cmp_1$ with specification $\cont_1$ that is not sufficient to implement $\cont$. Contracts provide an operation called \textit{quotient} which yields the specification whose implementation is \emph{exactly} the component $\cmp'$ such that $\cmp_1$ composed with $\cmp'$ meets the specification $\cont$. The operation of quotient has uses in synthesis (when we have made incremental progress towards meeting a goal) and in every situation where we need to find \textit{missing components}.

To say that contracts support concurrent design refers to another aspect of the design process. The design of some components involves multiple engineers working on different aspects of the same object. For example, a team may work on the functionality aspects of an integrated circuit, while another works on its timing characterization. If the functionality team generates a specification $\cont_f$, and the timing team generates a specification $\cont_t$, the two teams can combine their specs into a single contract object $\cont$ through an operation called \emph{merging}. In contract theory, these various aspects of a component are called \emph{viewpoints}.

The work on the assume-guarantee reasoning of Floyd-Hoare logic~\cite{Hoare:1969,floyd1967rw}, on assume-guarantee specifications of Lamport and Abadi~\cite{Lamport:1989:SAS:63238.63240,AbadiLamportComposingSpecs}, on design by contract by Meyer~\cite{meyerContract}, on interface automata by de Alfaro and Henzinger \cite{alfaroHenzingerIntAutomata},
and applications of formal specifications to cyber-physical systems by Damm~\cite{DammRichComponents},
yielded that formal assume-guarantee specifications could be used to design and analyze any cyber-physical system, including their discrete and continuous aspects. Assume-guarantee contracts were thus introduced for this purpose by Benveniste et al. in~\cite{multViewpoint}. Cyber-physical-system design methodologies using contracts were described in~\cite{SANGIOVANNIVINCENTELLI2012217,7268792}.

Two questions of practical relevance (the related discussions are written in {\color{blue}blue}) inform our discussion below.
\emph{\Question{Question~1: how can multiple specifications be combined to generate a system specification?}}
This leads us to consider binary operations on assume-guarantee contracts and their uses in the system design process.
The question has a corollary: \emph{\Question{when a system is decomposed across multiple suppliers, as well as across multiple viewpoints, is there a right order for applying the contract operations? Does the result depend on this order?}}
This question will lead us to consider how the various operations interact among themselves.

The second question has to do with computational complexity. We know that abstracting specifications tends to yield computationally-friendlier semidecision procedures.
\emph{\Question{Question~2: how can we compute contract abstractions?}} Part of our answer to this question will make use of semiring actions.

The structure of the paper is as follows.
Section~\ref{sc:agcont} begins to address Question~1 by covering the standard definitions of contracts~\cite{BenvenisteContractBook} and all known contract operations. The content of this section is a review, except for the operations of \textit{implication} and \textit{coimplication}, which, to the best of our knowledge, have not been published before. Section~\ref{sc:algstr} treats contracts as an algebra associated with any Boolean algebra and presents a study of monoid and semiring structures within a contract algebra. In other words, this section deals with how the various contract operations interact with each other, yielding insight into compositional design methodologies using contracts.
We use these results to define contract actions, which play a role in answering Question~2.
Section~\ref{sc:actions} covers two ways in which a Boolean algebra can act on its contract algebra, and Section~\ref{sc:contabs} discusses contract abstractions.
Sections~\ref{sc:algstr}-\ref{sc:contabs} summarize the contributions of this paper.

\section{Assume-guarantee contracts}
\label{sc:agcont}

This section defines assume guarantee contracts and addresses \Question{Question~1: how can multiple specifications be combined to generate a system specification?}

We discuss four binary operations that allow us to combine contracts. We also explore adjoint operations that allow us to carry out contract decompositions optimally. The content in this section borrows from \cite{BenvenisteContractBook,Negulescu95processspaces}, and the references in the text. Its contributions are the closed-form expressions for implication and coimplication and the use of duality to unify the presentation of the contract operations.

Let $\behaviorUnv$ be a set called the \emph{universe of behaviors}. Its elements are called \emph{behaviors}.  The universe of behaviors fixes the modeling formalism we have chosen in our application.
A property is defined as a subset of $\behaviorUnv$. A component is also a subset of $\behaviorUnv$. The difference is semantics: we think of components as the set of behaviors they can display. Properties contain behaviors meeting a certain criterion. We say that a component $M$ satisfies a property $P$, written $M \models P$ if $M \subseteq P$.
Assume-guarantee contracts are pairs of properties.

\begin{definition}
A contract $\cont$ is a pair of properties $\cont = (A, G)$. We call $A$ \emph{assumptions}, and $G$ \emph{guarantees}.
\end{definition}

Components can have two types of relationships with respect to a contract. 

\begin{definition}
Let $\cont = (A,G)$ be a contract. We say that a component $E$ is an environment for $\cont$, written $E \envsat \cont$, if $E \sat A$.
\end{definition}

Environments are those components which meet the assumptions of a contract. Implementations are those which meet the guarantees of the contract when operating in an environment accepted by the contract.

\begin{definition}
Let $\cont = (A,G)$ be a contract. We say that a component $\cmp$ is an implementation for $\cont$, written $\cmp \impsat \cont$, if $\cmp \cmpComposition E \sat G$ for every environment $E$ of $\cont$.
\end{definition}

Now that we have definitions for environments and implementations, we define a relation on contracts that declares two contracts equivalent when they have the same environments and the same implementations:

\begin{definition}
Let $\cont$ and $\cont'$ be two contracts. We say they are equivalent when they have the same environments and the same implementations.
\end{definition}

This means that for $\cont = (A, G)$ and $\cont' = (A', G')$ to be equivalent, we must have $A = A'$ and $G \setint A = G' \setint A' = G' \setint A$ (because $A' = A$). The largest $G'$ meeting this condition is $G' = G \setunion \neg A$ (the complement is taken with respect to $\behaviorset$). Enforcing this constraint for a contract allows us to have a unique mathematical object for each set of environments and implementations. We thus define an AG contract in canonical form as follows:
\begin{definition}
A contract in canonical form is a contract $\cont = (A, G)$ satisfying $A \setunion G = \behaviorset$.
\end{definition}

\textit{From now on, we assume all contracts are in canonical form.}

\subsection{Duality}

There is a unary operation which is helpful in revealing structure for AG contracts.

\begin{definition}
Let $\cont = (A, G)$ be a contract. We define a unary operation called reciprocal as follows: $C^{-1} = (G, A)$.
\end{definition}
This operation flips environments and implementations, i.e., it gives us the ``environment view'' of the specification $\cont$.
Note that the reciprocal respects canonicity.

\begin{definition}
Let $\circ$ and $\star$ be two binary operations on AG contracts, we say that the operations are dual when
$
    \left(\cont_a \circ \cont_b\right)^{-1} = \cont_a^{-1} \star \cont_b^{-1}
$.
\end{definition}

\subsection{Order}

\begin{definition}
Suppose $\cont$ and $\cont'$ are two contracts. We say that $\cont$ is a refinement of $\cont'$, written $\cont \le \cont'$, when all implementations of $\cont$ are implementations of $\cont'$ and all environments of $\cont'$ are environments of $\cont$.
\end{definition}

The association we make of a specification being a refinement is that it is harder to meet than another. This is why we say that a specification accepting more environments is a refinement of one accepting less. We can express this order relation using assumptions and guarantees.

\begin{proposition}[Theorem 5.2 of \cite{BenvenisteContractBook}]
Let $\cont = (A, G)$ and $\cont' = (A', G')$ be two contracts. Then $\cont \le \cont'$ when
$G \subseteq G' \text{ and } A' \subseteq A.$
\end{proposition}

\subsection{Conjunction and disjunction}

The notion of order provides a lattice structure to AG contracts in canonical form. Given contracts $\cont = (A, G)$ and $\cont' = (A', G')$, their meet (GLB) and join (LUB) are given by
\begin{align*}
    &\cont \meet \cont' = (A \setunion A', G \setint G') \quad \text{and} \quad 
    \cont \join \cont' = (A \setint A', G \setunion G').
\end{align*}

We leave it to the reader to verify that conjunction and disjunction are monotonic with respect to the refinement order. Also, conjunction and disjunction furnish our first example of dual operations:
$
  \cont \land \cont' = (G \setint G', A \setunion A' )^{-1}  = \left( (G, A) \lor (G', A') \right)^{-1} =
  \left(\cont^{-1} \lor \cont'^{-1}\right)^{-1}
$.


If we interpret a contract as the entailment $A\Rightarrow{G}$, then contract conjunction is the conjunction of such entailments. \Question{Conjunction can be used to combine viewpoints. Disjunction has an application in product lines.}

\subsection{Composition}

\Question{The notion of composition of AG contracts yields the specification of systems obtained from composing implementations of each of the contracts being composed. Contract composition formalizes how contracts with suppliers result in a system level contract.} This operation is defined by axiom as follows:

Suppose $\cont_1$ and $\cont_2$ are two specifications to be composed. Call $\cont$ the composite specification. Let $\cmp_1$ and $\cmp_2$ be arbitrary implementations of $\cont_1$ and $\cont_2$, respectively, and let $E$ be any environment of $\cont$. We define $\cont$ to be the smallest contract satisfying the following constraints:
the composite $\cmp_1 \cmpComposition \cmp_2$ is an implementation of $\cont$;
the composite $\cmp_1 \cmpComposition E$ is an environment of $\cont_2$; and
the composite $\cmp_2 \cmpComposition E$ is an environment of $\cont_1$.

The first requirement states that composing implementations of the specs being composed yields an implementation of the composite specification. The second requirement states that instantiating an implementation of $\cont_1$ in an environment of the composite specification yields an environment for $\cont_2$. And the last requirement is the analogous statement for $\cont_1$. This principle, which states how to compose specifications split between environment and implementation requirements, was stated for the first time by M. Abadi and L. Lamport~\cite{AbadiLamportComposingSpecs}.
We can obtain a closed-form expression of this principle for AG contracts:

\begin{proposition}[Theorem 5.2 of~\cite{BenvenisteContractBook}]
Let $\cont_1 = (A_1, G_1)$ and $\cont_2 = (A_2, G_2)$ be two AG contracts. Their composition, denoted $\cont_1 \parallel \cont_2$, is given by
$\cont_1 \parallel \cont_2 = \left( A_1 \setint A_2 \setunion \neg(G_1 \setint G_2) , G_1 \setint G_2\right)$.
\end{proposition}

We state without proof an important property of composition:

\begin{proposition}\label{{\documentPath}gbq78xrb6}
Composition of AG contracts is monotonic with respect to the refinement order.
\end{proposition}


\subsection{Strong merging (or merging)}

We said that AG contracts are used 
to handle the specifications of 
the various viewpoints of the same design element. Suppose $\cont_1$ and $\cont_2$ are specifications corresponding to different aspects to the same design object, e.g., functionality and power. We define their merger, denoted $\cont_1 \bullet \cont_2$, to be the contract which guarantees the guarantees of both specifications when the assumptions of both specifications are respected:
$
    \cont_1 \bullet \cont_2 = \left( A_1 \setint A_2, G_1 \setint G_2 \setunion \neg(A_1 \setint A_2) \right).
$
This contract is equivalent to  contract $\left( A_1 \setint A_2, G_1 \setint G_2\right)$, which is exactly what we defined merging to be. \Question{Merging can be used to combine viewpoints.}
Merging and composition are duals, as pointed out in \cite{Negulescu95processspaces}.

\subsection{Adjoints}

We have introduced four operations on AG contracts: two were obtained from the partial order, and two by axiom. Now we obtain the adjoints of these operations. Adjoints are used to compute optimal decompositions of contracts.

\subsubsection{Quotient (or residual)}
The adjoint of composition is called \emph{quotient}.
Let $\cont$ and $\cont'$ be two AG contracts. The quotient (also called residual in the literature), denoted $\cont / \cont'$, is defined as the largest AG contract $\cont''$ satisfying
$ \cont' \parallel \cont'' \le \cont$.

Due to the fact that the quotient is the largest contract with this property, Proposition~\ref{{\documentPath}gbq78xrb6} tells us that any of its refinements has this property. 

If we interpret $\cont$ as a top-level specification that our system has to meet (e.g., the specification of a vehicle), and $\cont'$ as the specification of a subset of the design for which we already have an implementation (e.g., a powertrain), then the quotient is the specification whose implementations are exactly those components that, if added to our partial design, would yield a system meeting the top-level specification. The following proposition gives us a closed-form expression for the quotient of AG contracts:

\begin{proposition}[Theorem 3.5 of~\cite{agquotient}]
Let $\cont = (A, G)$ and $\cont' = (A', G')$ be two AG contracts. The quotient, denoted $\cont / \cont'$, is given by
$$\cont / \cont' = \left( A \setint G', G \setint A' \setunion \neg (A \setint G')\right).$$
\end{proposition}

For an in-depth study of the notion of a quotient across several compositional theories, see~\cite{EPTCS326.14}. We can readily show that
\begin{align}\label{qixufksbk}
    \cont / \cont' = \cont \bullet (\cont')^{-1}.
\end{align}

\subsubsection{Separation}

Just like composition has an adjoint operation (the quotient), merging has an adjoint. For contracts $\cont$ and $\cont'$, we define the operation of \emph{separation}, denoted $\cont \div \cont'$, as the smallest contract $\cont''$ satisfying
$
\cont \le \cont' \bullet \cont''
$.
This operation has a closed-form solution:
\begin{proposition}[Theorem 3.12 of~\cite{contractMerging}]
Let $\cont = (A, G)$ and $\cont' = (A', G')$ be two AG contracts. Then
$\cont \div \cont' = \left( A \setint G' \setunion \neg (G \setint A'), G \setint A' \right)$.
\end{proposition}

Separation obeys
$
    \cont \div \cont' = \cont \parallel (\cont')^{-1}.
$
From this identity and~\eqref{qixufksbk}, it follows that quotient and separation are duals. For examples of merging and separation, see~\cite{contractMerging}.

\subsubsection{Implication and coimplication}

Given contracts $\cont$ and $\cont'$, the definition of implication, denoted $\cont'\to \cont$, in a lattice is $
\forall \cont''.\,\, \cont '' \land \cont'\le \cont \Leftrightarrow \cont'' \le (\cont'\to \cont)$.
In other words, $\cont'\to \cont$ is the largest contract $\cont''$ satisfying $\cont'' \land \cont'\le \cont$. The following proposition tells us how to compute this object:

\begin{proposition}\label{bkhlx}
    Let $\cont = (A, G)$ and $\cont'= (A', G')$ be two contracts. Implication has the closed form expression
    $
        \cont' \to \cont = \left( (A\setint \neg A') \setunion (G' \setint \neg G), G \setunion \neg G'\right)
    $.
\end{proposition}

Dually, we can ask what is the smallest contract $\cont''$ satisfying
$
\cont'' \lor \cont' \ge \cont
$.
We will denote this object $\cont' \nrightarrow \cont$. A similar proof yields the following proposition.
\begin{proposition}
    Let $\cont = (A, G)$ and $\cont'= (A', G')$ be two contracts. The smallest contract $\cont''$ satisfying $\cont'' \lor \cont' \ge \cont$ has the closed form expression
    $
        \cont' \nrightarrow \cont = \left( A \setunion \neg A', (G \setint \neg G') \setunion (A' \setint \neg A) \right)
    $.
\end{proposition}

We observe that
$$
\cont' \to \cont = \left( G \setunion \neg G', (A\setint \neg A') \setunion (G' \setint \neg G)\right)^{-1}  = \left( (\cont')^{-1} \nrightarrow \cont^{-1}\right)^{-1},
$$
which shows that implication and coimplication are duals.

\subsection{Summary of binary operations}

The following diagram shows how all AG contract operations are related.
{\scriptsize
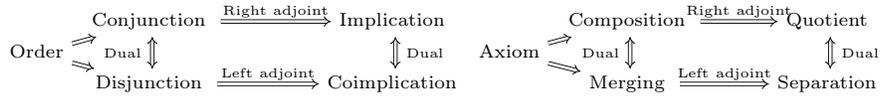
$$
{
{
    \begin{tikzcd}[row sep = 0.1em, column sep = tiny]
    & \text{Conjunction } 
    \arrow[ddd, "\text{Dual}\,\,"', Leftrightarrow] 
    \arrow[rrr, "\text{Right adjoint}", Rightarrow]
    & [2.9em] & & \text{{{Implication}} } 
    \\ 
    \text{Order} \arrow[ur, Rightarrow] \arrow[ddr, Rightarrow] & & \\
    \\
    & \text{Disjunction } 
    \arrow[rrr, "\text{Left adjoint}", Rightarrow] & & &
    \text{{{Coimplication}} } 
    \arrow[uuu, "\,\,\text{Dual}"', Leftrightarrow] 
    \end{tikzcd}
    \begin{tikzcd}[row sep = 0.1em, column sep = tiny]
    & \text{Composition } 
    \arrow[ddd, "\text{Dual}\,\,"', Leftrightarrow] 
    \arrow[rrr, "\text{Right adjoint}", Rightarrow]
    & [1.5em] & & \text{{{Quotient}} } 
    \\ 
    \text{Axiom} \arrow[ur, Rightarrow] \arrow[ddr, Rightarrow] & & \\
    \\
    & \text{{{Merging}} } 
    \arrow[rrr, "\text{Left adjoint}", Rightarrow] & & &
    \text{{{Separation}} } 
    \arrow[uuu, "\,\,\text{Dual}"', Leftrightarrow]
    \end{tikzcd}
}}$$}

Two operations---conjunction and disjunction---come from the definition of order, and two---composition and merging---are defined by axiom. The rest of the operations are adjoints of these four.

\section{Algebraic structures within contracts}
\label{sc:algstr}

\begin{table*}
    \centering
    {\footnotesize
    \begin{tabular}{l  l}
        \toprule
        Conjunction & Disjunction \\
        $\cont \meet \cont' = (a \lor a', g \land g')$ &
        $\cont \join \cont' = (a \land a', g \lor g')$ \\
        Composition  & Merging \\
        $\cont_1 \parallel \cont_2 = \left( a_1 \land a_2 \lor \neg(g_1 \land g_2) , g_1 \land g_2\right)$ &
        $\cont_1 \bullet \cont_2 = \left( a_1 \land a_2, g_1 \land g_2 \lor \neg(a_1 \land a_2) \right)$ \\
        Quotient & Separation \\
        $\cont / \cont' = \left( a \land g', g \land a' \lor \neg (a \land g')\right)$ &
        $\cont \div \cont' = \left( a \land g' \lor \neg (g \land a'), g \land a' \right)$ \\
        Implication & Coimplication \\
        $\cont' \to \cont = \left( (a\land \neg a') \lor (g' \land \neg g), g \lor \neg g'\right)$ &
        $\cont' \nrightarrow \cont = \left( a \lor \neg a', (g \land \neg g') \lor (a' \land \neg a) \right)$\\
        \bottomrule
    \end{tabular}}
    \caption{Closed-form expressions of operations for contracts over a Boolean algebra}
    \label{{\documentPath}kbdlsk}
    \vspace{-5mm}
\end{table*}

In this section, we investigate the corollary of \Question{Question 1: Since both viewpoints and subsystems specifications need to be combined, does the order among these operations influence the result? If yes, is there a best order?}

Inspection of the various binary formulas for AG contracts suggests that contracts can be defined over any Boolean algebra, not just that corresponding to properties over a set of behaviors. From now on, we deal with contracts defined over an arbitrary Boolean algebra and embark on a study of the relations between the various contract operations.

Let $\balg$ be a Boolean algebra with bottom and top elements $0_B$ and $1_B$, respectively. We form the contract algebra $\calg(\balg)$ associated with $\balg$. The elements of $\calg(\balg)$ are all pairs $(a, b) \in \balg^2$ such that $a \lor b = 1_B$. The notions of order and the binary operations work exactly the same as for AG contracts over sets of behaviors. Table~\ref{{\documentPath}kbdlsk} summarizes these operations.

The contract $1 = (0_B, 1_B)$ is larger than any contract. $0 = (1_B, 0_B)$ is smaller than any contract. The contract $\id = (1_B, 1_B)$ is an identity for composition and merging. $1$ is an identity for conjunction, and $0$ for disjunction.
Table~\ref{{\documentPath}nyfclfr} shows how various operations behave with respect to the distinguished elements.

\begin{table*}
    \centering
    {\footnotesize
    \begin{tabular}{l l l l }
        \toprule 
                      $0$ & $1$ & $\id$ \\
        \midrule
        $\cont \land 0 = 0$ & $\cont \land 1 = \cont$ & $(a,g) \land \id = (1_B, g)$\\
        $\cont \lor 0 = \cont$ & $\cont \lor 1 = 1$ & $(a,g) \lor \id = (a, 1_B)$\\
        $\cont \parallel 0 = 0$ & $(a,g) \parallel 1 = (\neg g, g)$ & $\cont \parallel \id = \cont$\\
        $(a,g) \bullet 0 = (a, \neg a)$ & $\cont \bullet 1 = 1$ & $\cont \bullet \id = \cont$\\
        \midrule
        $\cont / 0 = 1$ & $(a,g) / 1 = (a, \neg a)$ & $\cont / \id = \cont$ \\
        $0 / (a,g) = (g, \neg g)$ & $1 / \cont = 1$  & $\id / \cont = \cont^{-1}$ \\[0.2cm]
        $(a,g) \div 0 = (\neg g, g)$ & $\cont \div 1 = 0$ & $\cont \div \id = \cont$ \\
        $0 \div \cont = 0$ & $1 \div (a,g) = (\neg a, a)$ & $\id \div \cont = \cont^{-1}$ \\[0.2cm]
        $(a, g) \to 0 = (g, \neg g)$ & $\cont \to 1 = 1$ & $(a,g) \to \id = (\neg a, 1_B)$ \\
        $0 \to \cont = 1$ & $1 \to \cont = \cont$ & $\id \to (a,g) = (\neg g, g)$ \\[0.2cm]
        $\cont \nrightarrow 0 = 0$ & $(a,g) \nrightarrow 1 = (\neg a, a)$ & $(a,g) \nrightarrow \id = (1_B, \neg g)$ \\
        $0 \nrightarrow \cont = \cont$ & $1 \nrightarrow \cont = 0$ & $\id \nrightarrow (a, g) = (a, \neg a)$ \\
        \bottomrule
    \end{tabular}}
    \caption{Contract operations and the distinguished elements}
    \label{{\documentPath}nyfclfr}
\end{table*}
\subsection{Monoids}

We begin to study the interactions among the various operations.
We recall that a monoid is a semigroup with identity, i.e., a set together with an associative binary operation and an identity element for that operation. A contract algebra contains several monoids:

\begin{proposition}\label{kahfllnk}
    $\andmon(\balg) = (\calg(\balg), \land, 1_B)$, $\ormon(\balg) = (\calg(\balg), \lor, 0_B)$, $\parmon(\balg) = (\calg(\balg), \parallel, \id)$, and $\mermon(\balg) = (\calg(\balg), \bullet, \id)$ are idempotent, commutative monoids.
\end{proposition}

It turns out these monoids are isomorphic:
\newcommand{\rotsimeq}{\rotatebox[origin=c]{90}{$\simeq$}}
\newcommand{\rotsimeqm}{\rotatebox[origin=c]{-90}{$\simeq$}}

\begin{proposition}\label{nkjhlbskf}
    The monoids $\andmon(\balg)$, $\ormon(\balg)$, $\parmon(\balg)$, and $\mermon(\balg)$ are isomorphic. Moreover, the following diagram commutes, where $\theta_g (a, g) = (\neg(a \land g),g)$ and $\theta_a (a, g) = (a, \neg(a \land g))$:

    \begin{equation}\label{lkjvsfv}
        \begin{tikzcd}
            \andmon(\balg) \arrow[d, "\theta_g"', "\rotsimeqm", leftrightarrow]
            \arrow[r, "(\cdot)^{-1}", "\simeq"', leftrightarrow] 
            &
            \ormon(\balg)
            \arrow[d, "\theta_a", "\rotsimeq"', leftrightarrow]
            \\
            \parmon(\balg)
            \arrow[r, "(\cdot)^{-1}"', "\simeq", leftrightarrow] &
            \mermon(\balg)
        \end{tikzcd}
    \end{equation}
\end{proposition}

\subsection{Maps between monoids}

The previous result showed how to express contract operations in terms of others. Now we explore how to map contract monoids across their underlying Boolean algebras.

Suppose we have two Boolean algebras, $\balg$ and $\balg'$. Due to Proposition~\ref{nkjhlbskf}, it is sufficient to study the structure of the maps between the contract monoids $\parmon(\balg)$ and $\parmon(\balg')$ in order to understand the structure of the maps between all contract monoids associated with each Boolean algebra. First we study maps that allow us to construct and split contracts. Then we consider the general maps.


Let $\boolandmon(\balg)$ and $\boolormon(\balg)$ be the monoids $\boolandmon(\balg) = (\balg, \land, 1_B)$ and $\boolormon(\balg) = (\balg, \lor, 0_B)$. We define the two monoid maps
\begin{align*}
    &\iota_a\colon \boolandmon(\balg) \to \parmon(\balg) \quad a \mapsto (a, 1_B) \\
    &\iota_g\colon \boolandmon(\balg) \to \parmon(\balg) \quad g \mapsto (1_B, g).
\end{align*}

These maps generate an epic monoid map $\pi\colon \boolandmon(\balg) \times \boolandmon(\balg) \to \parmon(\balg)$ defined as
\[
    (a, g) \mapsto \iota_a(a) \parallel \iota_g(g) = (g \to a, g).
\]

Similarly, we have monoid maps that allow us to split a contract:
\begin{align*}
    &\pi_g \colon \parmon(\balg) \to \boolandmon(\balg) \quad \pi_g(a, g) = g \\
    &\pi_a \colon \andmon(\balg) \to \boolormon(\balg) \quad \pi_a(a,g) = a.
\end{align*}
We use the monoid isomorphisms~\eqref{lkjvsfv} to obtain a map $\andmon(\balg) \to \boolandmon(\balg)$ from the last morphism:
\begin{align*}
\neg \circ \pi_a \circ \theta_g \colon \parmon(\balg) & \to \boolandmon(\balg) \\
(a,g) & \mapsto a \land g.
\end{align*}

The two maps $\parmon(\balg) \to \boolandmon(\balg)$ yield the monic monoid map
\begin{align*}
\iota \colon \parmon(\balg) & \to \boolandmon(\balg) \times \boolandmon(\balg) \\
(a, g) & \mapsto (a \land g, g).
\end{align*}
This map is left-invertible:
\[
    \pi \circ \iota = \id.
\]

The elementary maps just described enable us to find the general structure between the monoid maps between the parallel monoids corresponding to two Boolean algebras.

\begin{theorem}\label{bnwjdnhx}
    Let $f\colon \parmon(\balg) \to \parmon(\balg')$. Then we can write $f$ as $$f = \pi \circ \left( l_a(ag) l_g(g) r_a(ag) r_g(g) , r_a(ag) r_g(g) \right) \circ \iota,$$where $l_a, l_g, r_a, r_g\colon \boolandmon(\balg) \to \boolandmon(\balg')$ are monoid morphisms.
\end{theorem}

\subsection{Semirings}

Now that we have four isomorphic monoids, we look for additional algebraic structure within the contract algebra, namely, the existence of semirings. This will capture the interactions between algebraic operations. 
First we study the distributivity of the binary operations. Distributivity and semi-distributivity answer the corollary of \Question{Question 1: Does the order between subspecification composition and viewpoint combination matter?}

\begin{table*}[ht!]
    \centering
    {\scriptsize
    \begin{tabular}{l | c c c c }
        \toprule
                      & Conjunction & Disjunction & Composition & Merging \\
        \midrule
        Conjunction   &
        $\begin{aligned} & \cont \land (\cont' \land \cont'') = \\ & (\cont \land \cont') \land (\cont \land \cont'')\end{aligned}$ &
        $\begin{aligned} & \cont \land (\cont' \lor \cont'') = \\ & (\cont \land \cont') \lor (\cont \land \cont'')\end{aligned}$ &
        $\begin{aligned} & \cont \land (\cont' \parallel \cont'') = \\ & (\cont \land \cont') \parallel (\cont \land \cont'')\end{aligned}$ &
        $\begin{aligned} & \id \land (1 \bullet 0) \ne \\ & (\id \land 1) \bullet (\id \land 0)\end{aligned}$
        \\[0.5cm]
        Disjunction   & 
        $\begin{aligned} & \cont \lor (\cont' \land \cont'') = \\ & (\cont \lor \cont') \land (\cont \lor \cont'')\end{aligned}$ &
        $\begin{aligned} & \cont \lor (\cont' \lor \cont'') = \\ & (\cont \lor \cont') \lor (\cont \lor \cont'')\end{aligned}$ &
        $\begin{aligned} & \id \lor (1 \parallel 0) \ne \\ & (\id \lor 1) \parallel (\id \lor 0)\end{aligned}$ &
        $\begin{aligned} & \cont \lor (\cont' \bullet \cont'') = \\ & (\cont \lor \cont') \bullet (\cont \lor \cont'')\end{aligned}$
        \\[0.5cm]
        Composition    &   
        $\begin{aligned} & \cont \parallel (\cont' \land \cont'') = \\ & (\cont \parallel \cont') \land (\cont \parallel \cont'')\end{aligned}$ &
        $\begin{aligned} & \cont \parallel (\cont' \lor \cont'') = \\ & (\cont \parallel \cont') \lor (\cont \parallel \cont'')\end{aligned}$ &
        $\begin{aligned} & \cont \parallel (\cont' \parallel \cont'') = \\ & (\cont \parallel \cont') \parallel (\cont \parallel \cont'')\end{aligned}$ &
        $\begin{aligned} & 1 \parallel (0 \bullet \id) \ne \\ & (1 \parallel 0) \bullet (1 \parallel \id)\end{aligned}$
        \\[0.5cm]
        Merging       & 
        $\begin{aligned} & \cont \bullet (\cont' \land \cont'') = \\ & (\cont \bullet \cont') \land (\cont \bullet \cont'')\end{aligned}$ &
        $\begin{aligned} & \cont \bullet (\cont' \lor \cont'') = \\ & (\cont \bullet \cont') \lor (\cont \bullet \cont'')\end{aligned}$ &
        $\begin{aligned} & 0 \bullet (1 \parallel \id) \ne \\ & (0 \bullet 1) \parallel (0 \bullet \id)\end{aligned}$ &
        $\begin{aligned} & \cont \bullet (\cont' \bullet \cont'') = \\ & (\cont \bullet \cont') \bullet (\cont \bullet \cont'')\end{aligned}$
        \\   
        \bottomrule
    \end{tabular}}
    \caption{Distributivity of contract operations}
    \label{{\documentPath}kcsdy}
\end{table*}

\begin{proposition}\label{qocvksf}
    Table~\ref{{\documentPath}kcsdy} shows whether the binary operations displayed in the rows distribute over the binary operations in the columns.
\end{proposition}


If we use conjunction to combine viewpoints, then the distributivity of composition over conjunction says that \Question{combining viewpoints and composing subspecifications can be performed at will, in any order.} This is wrong if we use merging to combine viewpoints. Nevertheless, the following proposition expresses that, if we use merging to combine viewpoints, \Question{it is preferred to combine viewpoints first, and then compose subspecifications, than the converse.}

\begin{proposition}\label{nkkbhgfb}
Let $\cont_i$ be a contract for $1 \le i \le 4$. We have the following semi-distributions:
\begin{align*}
(C_1 \land C_2) &\star (C_3 \land C_4) \le (C_1 \star C_3) \land (C_2 \star C_4)
\quad\text{and} \\
(C_1 \lor C_2) & \star (C_3 \lor C_4) \ge (C_1 \star C_3) \lor (C_2 \star C_4),
\end{align*}
where $\star$ is any of the operations conjunction, disjunction, composition, or merging.
\end{proposition}

The sub-distributivity of parallel composition over conjunction is proved in Chapter 4 of \cite{BenvenisteContractBook}.
We will now use the distributivity results to look for semiring structure within the algebra of contracts. We recall the definition of a semiring (see, e.g.,~\cite{semiringsGolan}):

\begin{definition}
    A semiring $(R,\cdot, + , 1_R, 0_R)$ is a nonempty set $R$ where
    (a) $(R, +, \allowbreak 0_R)$ is a commutative monoid;
    (b) $(R, \cdot, 1_R)$ is a monoid;
    (c) $r(s + t) = rs + rt$ and $(s+t)r = sr + tr$ for all $r, s, t \in R$;
    (d) $r \cdot 0_R = 0_R \cdot r = 0_R$ for all $r \in R$; and
    (e) $0_R \ne 1_R$.
    
    A map of semirings $f: (R,\cdot, + , 1_R, 0_R) \to (R',\cdot, + , 1_{R'}, 0_{R'})$ satisfies
    (a) $f(0_R) = f(0_{R'})$,
    (b) $f(1_R) = f(1_{R'})$,
    (c) $f(r + s) = f(r) + f(s)$, and
    (d) $f(r \cdot s) = f(r) \cdot f(s)$.
\end{definition}

The following result provides the semiring structures available in the contract algebra.

\begin{proposition}\label{vfglgqbwf}
    Using the operations of conjunction, disjunction, composition, and merging, we have exactly four semirings:
    (a) the conjunction semiring $\andsr(\balg) = (\calg(\balg), \land, \lor, 1, 0)$,
    (b) the disjunction semiring $\orsr(\balg) = (\calg(\balg), \lor, \land, \allowbreak 0, 1)$,
    (c) the composition semiring $\parsr(\balg) = (\calg(\balg), \parallel, \lor, \id, 0)$, and
    (d) the merging semiring $\mersr(\balg) = (\calg(\balg), \bullet, \land, \id, 1)$.
\end{proposition}

These four semirings have two isomorphisms.
\begin{proposition}\label{nhwlfchcjbg}
We have the isomorphisms $\andsr(\balg) \cong \orsr(\balg)$ and $\parsr(\balg) \cong \mersr(\balg)$. There are no isomorphisms between these two pairs.
\end{proposition}

\subsection{Some semiring maps}

Let $\boolandsr(\balg)$ and $\boolorsr(\balg)$ be, respectively, the semirings $(\balg, \land, \lor, 1, 0)$ and $(\balg, \lor, \land, \allowbreak 0, 1)$. We first observe that complementation is a semiring isomorphism for $\boolandsr(\balg)$ and $\boolorsr(\balg)$.
We define maps $\Delta_g : \boolandsr(\balg) \to \andsr(\balg)$ and $\iota_g : \boolandsr(\balg) \to \parsr(\balg)$ as follows:
$
    \Delta_g (b) = (\neg b, b)$ and
$
    \iota_g (b) = (1_B, b)
$.
\begin{proposition}\label{khjbcgfkbgfk}
    $\Delta_g$ and $\iota_g$ are semiring homomorphisms. 
\end{proposition}
Observe that $\Delta_g$ can be used to obtain a semiring map from $\boolorsr(\balg)$ to $\orsr(\balg)$ using the semiring isomorphisms
$\neg: \boolandsr(\balg) \xrightarrow{\sim} \boolorsr(\balg)$ and
$(\cdot)^{-1}: \andsr(\balg) \xrightarrow{\sim} \orsr(\balg)$ as follows:
\begin{equation*}
    \begin{tikzcd}[row sep=tiny]
        \boolorsr(\balg) \arrow[r, "\neg"] & \boolandsr(\balg) \arrow[r, "\Delta_g"] & \andsr(\balg) \arrow[r, "(\cdot)^{-1}"] & \orsr(\balg) \\
        b \arrow[r, mapsto] & \neg b \arrow[r, mapsto] & (b, \neg b) \arrow[r, mapsto] & (\neg b, b)
    \end{tikzcd}
\end{equation*}
This means that $\Delta_g$ is also a semiring homomorphism $\begin{tikzcd}
    \boolorsr(\balg) \arrow[r, "\Delta_g"] & \orsr(\balg).
\end{tikzcd}$ The following diagram commutes in the category of semirings:
\begin{equation*}
    \begin{tikzcd}
        \boolandsr(\balg) \arrow[d, "\Delta_g"'] \arrow[r, "\neg", "\simeq"', leftrightarrow] 
        \arrow[dr, dashed]
        &
        \boolorsr(\balg)
        \arrow[d, "\Delta_g"]
        \arrow[dl, "\Delta_a" description, dashed]
        \\
        \andsr(\balg) \arrow[r, "(\cdot)^{-1}"', "\simeq", leftrightarrow] & \orsr(\balg)
    \end{tikzcd}
\end{equation*}
The commutativity of the diagram gives rise to the diagonal arrow $\Delta_a = (\cdot)^{-1} \circ \Delta_g = \Delta_g \circ \neg$. This map is a semiring homomorphism from $\boolandsr(\balg)$ to $\orsr(\balg)$ and from $\boolorsr(\balg)$ to $\andsr(\balg)$. Explicitly, this map is
\[
    \Delta_a b = (\Delta_g (b))^{-1} = (\neg b, b)^{-1} = (b, \neg b) \quad (b \in \balg).
\]
If we use the map $\iota_g$, we can obtain $\iota_a': \boolorsr(\balg) \to \mersr(\balg)$ as follows:
\begin{equation*}
    \begin{tikzcd}[row sep=tiny]
        \boolorsr(\balg) \arrow[r, "\neg"] & \boolandsr(\balg) \arrow[r, "\iota_g"] & \parsr(\balg) \arrow[r, "(\cdot)^{-1}"] & \orsr(\balg) \\
        b \arrow[r, mapsto] & \neg b \arrow[r, mapsto] & (1_B, \neg b) \arrow[r, mapsto] & (\neg b, 1_B)
    \end{tikzcd}
\end{equation*}
We obtain the diagram below.
\begin{equation*}
    \begin{tikzcd}[row sep=large]
        \boolandsr(\balg) \arrow[d, "\iota_g"'] \arrow[r, "\neg", "\simeq"', leftrightarrow] 
        \arrow[dr, "\iota_a"', pos=0.1,dashed]
        &
        \boolorsr(\balg)
        \arrow[d, "\iota_a'"]
        \arrow[dl, "\iota_g'", pos=0.1, dashed]
        \\
        \parsr(\balg) \arrow[r, "(\cdot)^{-1}"', "\simeq", leftrightarrow] & \mersr(\balg)
    \end{tikzcd}
\end{equation*}
The commutativity of the diagram provides the semiring maps
$\begin{tikzcd}[column sep=small]
    \boolandsr(\balg) \arrow[r, "\iota_a"] & \mersr(\balg)
\end{tikzcd}
$
and
$
\begin{tikzcd}[column sep=small]
    \boolorsr(\balg) \arrow[r, "\iota_g'"] & \parsr(\balg)
\end{tikzcd}$
given by
$\iota_a = (\cdot)^{-1} \circ \iota_g$ and $\iota_g' = \iota_g \circ \neg$.

Now we consider maps to $\boolandsr(\balg)$. The map
$\pi_g: (a, g) \mapsto g$
is a semiring homomorphism from $\andsr(\balg)$ to $\boolandsr(\balg)$ and from $\parsr(\balg)$ to $\boolandsr(\balg)$. Similarly, the map 
$\pi_a': (a, g) \mapsto \neg a$
is a semiring homomorphism from $\andsr(\balg)$ to $\boolandsr(\balg)$. The following diagrams commute and define the maps not specified before.
\begin{equation*}
    \begin{tikzcd}[row sep=large]
        \boolandsr(\balg) \arrow[d, "\pi_g"', leftarrow] \arrow[r, "\neg", "\simeq"', leftrightarrow] 
        \arrow[dr, "\pi_a", leftarrow, pos=0.90,dashed]
        &
        \boolorsr(\balg)
        \arrow[d, "\pi_a'", leftarrow]
        \arrow[dl, "\pi_g'"', pos=0.90, dashed, leftarrow]
        \\
        \parsr(\balg) \arrow[r, "(\cdot)^{-1}"', "\simeq", leftrightarrow] & \mersr(\balg)
    \end{tikzcd}
    \quad
    \begin{tikzcd}[row sep=large]
        \boolandsr(\balg) \arrow[d, "\pi_g"', leftarrow] \arrow[r, "\neg", "\simeq"', leftrightarrow] 
        \arrow[dr, "\pi_a", leftarrow, pos=0.90,dashed]
        &
        \boolorsr(\balg)
        \arrow[d, "\pi_a'", leftarrow]
        \arrow[dl, "\pi_g'"', pos=0.90, dashed, leftarrow]
        \\
        \andsr(\balg) \arrow[r, "(\cdot)^{-1}"', "\simeq", leftrightarrow] & \orsr(\balg)
    \end{tikzcd}
    \quad
    \begin{tikzcd}[row sep=large]
        \boolandsr(\balg) \arrow[d, "\pi_a'"', leftarrow] \arrow[r, "\neg", "\simeq"', leftrightarrow] 
        \arrow[dr, "\pi_g'", leftarrow, pos=0.90,dashed]
        &
        \boolorsr(\balg)
        \arrow[d, "\pi_g", leftarrow]
        \arrow[dl, "\pi_a"', pos=0.90, dashed, leftarrow]
        \\
        \andsr(\balg) \arrow[r, "(\cdot)^{-1}"', "\simeq", leftrightarrow] & \orsr(\balg)
    \end{tikzcd}
\end{equation*}
\section{Actions}
\label{sc:actions}

One of the questions we sought to answer was: how can we compute contract abstractions? Abstractions are useful to carry out refinement verification in less costly computational environments. We will define abstractions through contract actions.

The semiring maps just described can be used to generate actions of the semirings $\boolandsr(\balg)$ and $\boolorsr(\balg)$ over the contract semirings. Consider, for example the map $\Delta_g : \boolandsr(\balg) \to \andsr(\balg)$. For a contract $\cont = (a,g)$, we have
$
    \Delta_g(b) \land \cont = (\neg b, b) \land (a, g) = (b \to a, b \land g)
$.

Now consider the map $\iota_g : \boolandsr(\balg) \to \parsr(\balg)$:
$
    \iota_g(b) \parallel \cont = (1_B, b) \parallel (a, g) = ((b \land g) \to a, b \land g) = (b \to a, b \land g)
$.
We observe that $\boolandsr(\balg)$ acts in the same way on the semirings $\andsr(\balg)$ and $\parsr(\balg)$. This leads us to
\begin{definition}
The right action of $\balg$ on $\calg(\balg)$ is
$
    (a, g) \cdot b = (b \to a, b \land g)
$.
\end{definition}

Similarly, the semiring map $\Delta_a : \boolandsr(\balg) \to \orsr(\balg)$ yields
$
    \Delta_a(b) \lor \cont = (b, \neg b) \lor (a, g) = (b \land a, b \to g)
$,
and the semiring map $\iota_a : \boolandsr(\balg) \to \mersr(\balg)$ results in
$
    \iota_a(b) \bullet \cont = (b, 1_B) \bullet (a, g) = (b \land a, b \to g)
$.
$\boolandsr(\balg)$ acts in the same way on the semirings $\orsr(\balg)$ and $\mersr(\balg)$.
We thus obtain
\begin{definition}
The left action of $\balg$ on $\calg(\balg)$ is given by
$
    b \cdot (a, g) = (b \land a, b \to g)
$.
\end{definition}

The left and right actions have the practical meaning of adding assumptions or guarantees, respectively, to a contract.
The following proposition shows several properties of these actions.

\setlength\tabcolsep{4pt}
\begin{table*}
    \centering
    {\scriptsize
    \begin{tabular}{l  l}
        \toprule
        Order & \\
        $b \cdot \cont \ge \cont$ &
        $\cont \cdot b \le \cont$ \\
        $\cont \le \cont' \Rightarrow b \cdot \cont \le b\cdot \cont'$ &
        $\cont \le \cont' \Rightarrow \cont \cdot b \le \cont' \cdot b$ \\
        \midrule
        Reciprocal & \\
        $(b \cdot \cont)^{-1} = \cont^{-1} \cdot b$ & \\
        \midrule
        Associativity & \\
        $(b \land b')\cdot \cont = b \cdot (b' \cdot \cont)$ &
        $\cont\cdot (b \land b') = (\cont \cdot b) \cdot b'$ \\
        \midrule
        Distributivity over the Boolean algebra &  \\
        $(b \lor b') \cdot \cont = (b \cdot \cont) \land (b' \cdot \cont)$ &
        $\cont \cdot (b \lor b') = (\cont \cdot b) \lor (\cont \cdot b')$
        \\
        \midrule
        Actions and the contract operations \\
        $b\cdot (\cont \land \cont') = b\cdot \cont \land b\cdot \cont'$ &
        $(\cont \land \cont')\cdot b = \cont \cdot b \land \cont'$ \\
        $b\cdot (\cont \lor \cont') = b \cdot \cont \lor \cont'$ &
        $(\cont \lor \cont')\cdot b = \cont \cdot b \lor \cont' \cdot b$ \\
        $b\cdot (\cont \parallel \cont') = b \cdot \cont \parallel b\cdot \cont'$ &
        $(\cont \parallel \cont')\cdot b = \cont \cdot b \parallel \cont'$ \\
        $b\cdot (\cont \bullet \cont') = b \cdot \cont \bullet \cont'$ &
        $(\cont \bullet \cont')\cdot b = (\cont \cdot b) \bullet (\cont' \cdot b)$ \\
        \midrule
        Actions and the adjoint operations \\
        $b\cdot (\cont / \cont') = \cont / (\cont' \cdot b) = (b \cdot \cont) / \cont'$ &
        $(\cont / \cont')\cdot b = (\cont \cdot b) / (b \cdot \cont')$ \\
        $b\cdot (\cont \div \cont') = (b \cdot \cont) \div (\cont' \cdot b)$ &
        $(\cont \div \cont')\cdot b = (\cont \cdot b) \div \cont' = \cont \div (b \cdot \cont')$ \\
        $b\cdot (\cont' \to \cont) = \cont' \to b \cdot \cont = \cont' \cdot b \to \cont$ &
        $(\cont' \to \cont)\cdot b = b\cdot \cont' \to \cont \cdot b$ \\
        $b\cdot (\cont' \nrightarrow \cont) = \cont' \cdot b \nrightarrow b \cdot \cont$ &
        $(\cont' \nrightarrow \cont)\cdot b = \cont'  \nrightarrow \cont \cdot b = b \cdot \cont' \nrightarrow \cont$
        \\
        \bottomrule
    \end{tabular}}
    \caption[Identities for the left and right actions of a Boolean algebra $\balg$ over its contract algebra]{Identities for the left and right actions of a Boolean algebra $\balg$ over its contract algebra $(b, b' \in \balg \text{ and } \cont, \cont' \in \calg(\balg))$}
    \label{kslflp}
\end{table*}

\begin{proposition}\label{knxfghsqblr}
    The identities for the left and right actions shown in Table~\ref{kslflp} hold.
\end{proposition}


\Question{The left action on contracts can be used to generate an abstraction for the guarantees of a contract}. Given a contract $\cont$ and $b \in B$, we know from Table~\ref{kslflp} that $\cont \le b \cdot \cont$. This tells us that we can obtain a more relaxed contract by adding assumptions.
Subsequently, we may make use of this additional assumption to coarsen the guarantees of the contract.
This operation is rich in algebraic properties, as shown in Table~\ref{kslflp}.
For algorithmic manipulations of contracts, see Chapter 7 of \cite{Incer:EECS-2022-99}.

\section{Contract abstractions}
\label{sc:contabs}

It is often useful to vary the level of detail used to represent objects. More detail may be needed to carry out some analysis tasks, but too much detail may hinder computation. We will explore \Question{Question~2: how can we compute contract abstractions?}

\subsection{The Galois-connection abstraction}
Suppose we have a monotone operator $\alpha\colon \balg_c \to \balg_a$, where $B_a$ and $B_c$ are Boolean algebras called abstract and concrete domains, respectively. Suppose that $\alpha$ is also a monoid map $\alpha \colon \boolormon(B_c) \to \boolormon(B_a)$, i.e., it commutes with disjunction and maps $1_{B_c}$ to $1_{B_a}$.
We can define a map $\bar \alpha \colon \calg(B_c) \to \calg(B_a)$ as
\begin{equation}\label{qlkjvbhqk}\bar \alpha\colon (a, g) \mapsto (\alpha (a), \alpha (g)).\end{equation} The map is well-defined since $\alpha (a) \lor \alpha (g) = \alpha(a \lor g) = \alpha(1_{\balg_c}) = 1_{\balg_a}$. Moreover, the monotonicity of $\bar \alpha$ follows from the monotonicity of $\alpha$. This abstraction is a slight generalization of that proposed in Chapter 5 of \cite{BenvenisteContractBook}. 

How can such a map $\alpha$ be obtained?
Suppose that $\alpha$ is a monotone map that maps $1_{B_c}$ to $1_{B_a}$.
If $\alpha$ is a left element of a Galois connection pair, then it commutes with disjunction (because left adjoints commute with colimits). Thus, such an $\alpha$ generates a valid contract abstraction \eqref{qlkjvbhqk}.

\subsection{Contract Galois connections from Boolean algebra maps}

Let $f: \balg \to \balg'$ be a Boolean algebra map. $f$ induces a map $f^*: \calg(\balg) \to \calg(\balg')$ between their contract algebras given by
$
    f^*(a, g) = (f(a), f(g))
$.
Observe that $f(a) \lor f(g) = f(a \lor g) = f(1_{\balg}) = 1_{\balg'}$. As $f$ commutes with the Boolean algebra operations, $f^*$ commutes with the contract operations. Thus, $f^*$ is well-defined.

Let $\gamma: B_a \to B_c$ and $\alpha: B_c \to B_a$ be Boolean algebra maps. These maps generate contract maps
$
    \begin{tikzcd}
        \calg(B_a) \arrow[r, "\gamma^*"] & \calg(B_c)
    \end{tikzcd}
$
and
$
    \begin{tikzcd}
        \calg(B_c) \arrow[r, "\alpha^*"] & \calg(B_a),
    \end{tikzcd}
$
as described before.

We are interested in exploring the conditions needed for these maps to form a Galois connection. Specifically, for $\cont_a = (a_a, g_a) \in \calg(B_a)$ and $\cont_c = (a_c, g_c) \in \calg(B_c)$, we want
$
    \alpha^*(\cont_c) \le \cont_a \text{ if and only if }
    \cont_c \le \gamma^*(\cont_a)
$.

This means that
$
    (\alpha a_c, \alpha g_c) \le (a_a, g_a) \text{ if and only if }
    (a_c, g_c) \le (\gamma a_a, \gamma g_a).
$
If we set $a_c = 1_{B_c}$ and $a_a = 1_{B_a}$, we get
$
    \alpha g_c \le g_a \text{ if and only if }
    g_c \le \gamma g_a
$,
and if we set $g_c = 1_{B_c}$ and $g_a = 1_{B_a}$, we obtain
$
    a_a \le \alpha a_c \text{ if and only if }
    \gamma a_a \le a_c
$.

This means that $\alpha$ and $\gamma$ must be simultaneously the left and right adjoints of each other. By setting $g_c = \gamma g_a$ and $a_c = \gamma a_a$ in the equations above, we obtain that $\alpha \circ \gamma$ is the identity map. Similarly, by setting $g_a = \alpha g_c$ and $a_a = \alpha a_c$, we get that $\gamma \circ \alpha$ is the identity map. This means that $B_a$ and $B_c$ are isomorphic. We conclude that contract Galois connections generated from Boolean algebra maps impose very rigid constraints on the Boolean algebras over which the contracts are defined.

\section{Conclusions}
\vspace{-2mm}

Assume-guarantee contracts 
provide effective tools 
in system engineering design. In this paper, we described the algebra of contracts that can provide a framework for manipulating contracts in a structured way. 

We explored in-depth this algebraic structure \textit{per se}, while establishing a mathematically sound methodology for contract-based system design. The rich algebraic structure of assume-guarantee contracts provides sound support for a number of operations in system design: combining viewpoints, composing sub-specifications, and patching a design. At the same time, it raises a number of open methodological issues: does a designer have freedom in ordering the different design steps? Alternatively, does the theory impose or recommend an ordering? Abstractions are a well established way to provide semi-decision procedures at a lower computational cost. In this paper, we answered the following question: is there a systematic way to lift abstractions of properties to abstractions of contracts? 

While writing assume-guarantee contracts is intuitive in applications, the computation of the algebraic operations is often not.
Pacti\footnote{\url{https://www.pacti.org}} was recently introduced \cite{incer2023pacti}
to support the design tasks that are backed by the algebra of contracts. Pacti is able to automatically compute several of the operations we discussed.

An important question remains open: is there a sound way to lift testing from properties---where this is well-developed---to assume-guarantee contracts? Our algebra of assume-guarantee contracts does not provide answers to this question as yet.


\bibliographystyle{style/lncs/splncs04}
\bibliography{support/definitions,support/references}

\begin{thebibliography}{10}
\providecommand{\url}[1]{\texttt{#1}}
\providecommand{\urlprefix}{URL }
\providecommand{\doi}[1]{https://doi.org/#1}

\bibitem{AbadiLamportComposingSpecs}
Abadi, M., Lamport, L.: Composing specifications. ACM Transactions on
  Programming Languages and Systems  \textbf{15}(1),  73--132 (Jan 1993).
  \doi{10.1145/151646.151649}, \url{http://doi.acm.org/10.1145/151646.151649}

\bibitem{alfaroHenzingerIntAutomata}
de~Alfaro, L., Henzinger, T.A.: Interface automata. In: Proceedings of the 8th
  European Software Engineering Conference Held Jointly with 9th ACM SIGSOFT
  International Symposium on Foundations of Software Engineering. pp. 109--120.
  ESEC/FSE-9, Association for Computing Machinery, New York, NY, USA (2001).
  \doi{10.1145/503209.503226}, \url{https://doi.org/10.1145/503209.503226}

\bibitem{multViewpoint}
Benveniste, A., Caillaud, B., Ferrari, A., Mangeruca, L., Passerone, R.,
  Sofronis, C.: Multiple viewpoint contract-based specification and design. In:
  {de Boer}, F.S., Bonsangue, M.M., Graf, S., {Willem-Paul de Roever} (eds.)
  Formal Methods for Components and Objects, $6^{th}$ International Symposium
  (FMCO 2007), Amsterdam, The Netherlands, October 24--26, 2007, Revised
  Papers, Lecture Notes in Computer Science, vol.~5382, pp. 200--225. Springer
  Verlag, Berlin Heidelberg (2008). \doi{10.1007/978-3-540-92188-2}

\bibitem{BenvenisteContractBook}
Benveniste, A., Caillaud, B., Nickovic, D., Passerone, R., Raclet, J.B.,
  Reinkemeier, P., Sangiovanni-Vincentelli, A.L., Damm, W., Henzinger, T.A.,
  Larsen, K.G.: Contracts for system design. Foundations and
  Trends$^{\text{\scriptsize{\textregistered}}}$\hspace{-.3em} in Electronic
  Design Automation  \textbf{12}(2-3),  124--400 (2018)

\bibitem{DammRichComponents}
{Damm}, W.: Controlling speculative design processes using rich component
  models. In: Fifth International Conference on Application of Concurrency to
  System Design (ACSD'05). pp. 118--119 (June 2005). \doi{10.1109/ACSD.2005.35}

\bibitem{davies2007organizing}
Davies, A., Brady, T., Hobday, M.: Organizing for solutions: Systems seller vs.
  systems integrator. Industrial marketing management  \textbf{36}(2),
  183--193 (2007)

\bibitem{https://doi.org/10.1002/sys.20028}
Emes, M., Smith, A., Cowper, D.: Confronting an identity crisis—how to
  “brand” systems engineering. Systems Engineering  \textbf{8}(2),
  164--186 (2005). \doi{https://doi.org/10.1002/sys.20028},
  \url{https://onlinelibrary.wiley.com/doi/abs/10.1002/sys.20028}

\bibitem{floyd1967rw}
Floyd, R.W.: Assigning meanings to programs. Mathematical Aspects of Computer
  Science, ed. Schwartz, JT, Amer. Math. Soc  (1967)

\bibitem{semiringsGolan}
Golan, J.S.: Semirings and their Applications. Springer, Dordrecht, 1st ed edn.
  (1999). \doi{10.1007/978-94-015-9333-5}

\bibitem{Hoare:1969}
Hoare, C.A.R.: An axiomatic basis for computer programming. Commun. ACM
  \textbf{12}(10),  576--580 (Oct 1969). \doi{10.1145/363235.363259},
  \url{http://doi.acm.org/10.1145/363235.363259}

\bibitem{hobday2005systems}
Hobday, M., Davies, A., Prencipe, A.: Systems integration: a core capability of
  the modern corporation. Industrial and corporate change  \textbf{14}(6),
  1109--1143 (2005)

\bibitem{iannopollo2014library}
Iannopollo, A., Nuzzo, P., Tripakis, S., Sangiovanni-Vincentelli, A.:
  Library-based scalable refinement checking for contract-based design. In:
  2014 Design, Automation \& Test in Europe Conference \& Exhibition (DATE).
  pp.~1--6. IEEE (2014)

\bibitem{agquotient}
Incer, I., Sangiovanni-Vincentelli, A.L., Lin, C.W., Kang, E.: Quotient for
  assume-guarantee contracts. In: 16th ACM-IEEE International Conference on
  Formal Methods and Models for System Design. pp. 67--77. MEMOCODE'18 (October
  2018). \doi{10.1109/MEMCOD.2018.8556872}

\bibitem{Incer:EECS-2022-99}
Incer, I.: The Algebra of Contracts. Ph.D. thesis, EECS Department, University
  of California, Berkeley (May 2022)

\bibitem{incer2023pacti}
Incer, I., Badithela, A., Graebener, J., Mallozzi, P., Pandey, A., Yu, S.J.,
  Benveniste, A., Caillaud, B., Murray, R.M., Sangiovanni-Vincentelli, A.,
  et~al.: Pacti: Scaling assume-guarantee reasoning for system analysis and
  design. arXiv preprint arXiv:2303.17751  (2023)

\bibitem{EPTCS326.14}
Incer, I., Mangeruca, L., Villa, T., Sangiovanni-Vincentelli, A.L.: The
  quotient in preorder theories. In: Raskin, J.F., Bresolin, D. (eds.)
  {Proceedings 11th International Symposium on} Games, Automata, Logics, and
  Formal Verification, {Brussels, Belgium, September 21-22, 2020}. Electronic
  Proceedings in Theoretical Computer Science, vol.~326, pp. 216--233. Open
  Publishing Association, Brussels, Belgium (2020). \doi{10.4204/EPTCS.326.14}

\bibitem{https://doi.org/10.1002/inst.201013143}
Jackson, S.: Memo to industry: The crisis in systems engineering. INSIGHT
  \textbf{13}(1),  43--43 (2010). \doi{https://doi.org/10.1002/inst.201013143},
  \url{https://onlinelibrary.wiley.com/doi/abs/10.1002/inst.201013143}

\bibitem{Lamport:1989:SAS:63238.63240}
Lamport, L.: A simple approach to specifying concurrent systems. Commun. ACM
  \textbf{32}(1),  32--45 (Jan 1989). \doi{10.1145/63238.63240},
  \url{http://doi.acm.org/10.1145/63238.63240}

\bibitem{LEVESON2004237}
Leveson, N.: A new accident model for engineering safer systems. Safety Science
   \textbf{42}(4),  237--270 (2004).
  \doi{https://doi.org/10.1016/S0925-7535(03)00047-X},
  \url{https://www.sciencedirect.com/science/article/pii/S092575350300047X}

\bibitem{https://doi.org/10.1002/sys.21249}
Madni, A.M., Sievers, M.: Systems integration: Key perspectives, experiences,
  and challenges. Systems Engineering  \textbf{17}(1),  37--51 (2014).
  \doi{https://doi.org/10.1002/sys.21249},
  \url{https://onlinelibrary.wiley.com/doi/abs/10.1002/sys.21249}

\bibitem{10.1007/978-3-030-89247-0_3}
Mallozzi, P., Nuzzo, P., Pelliccione, P.: Incremental refinement of goal models
  with contracts. In: Hojjat, H., Massink, M. (eds.) Fundamentals of Software
  Engineering. pp. 35--50. Springer International Publishing, Cham (2021)

\bibitem{meyerContract}
Meyer, B.: Applying ``design by contract''. IEEE Computer  \textbf{25}(10),
  40--51 (October 1992). \doi{10.1109/2.161279}

\bibitem{national2016top}
{National Defense Industrial Association and others}: Top systems engineering
  issues in {US} defense industry. Systems Engineering Division Task Group
  Report  (2016)

\bibitem{Negulescu95processspaces}
Negulescu, R.: Process spacess. Tech. Rep. CS-95-48, University of Waterloo
  (1995)

\bibitem{7268792}
{Nuzzo}, P., {Sangiovanni-Vincentelli}, A.L., {Bresolin}, D., {Geretti}, L.,
  {Villa}, T.: A platform-based design methodology with contracts and related
  tools for the design of cyber-physical systems. Proceedings of the IEEE
  \textbf{103}(11),  2104--2132 (2015). \doi{10.1109/JPROC.2015.2453253}

\bibitem{contractMerging}
Passerone, R., Incer, I., Sangiovanni-Vincentelli, A.L.: Coherent extension,
  composition, and merging operators in contract models for system design. ACM
  Trans. Embed. Comput. Syst.  \textbf{18}(5s) (Oct 2019).
  \doi{10.1145/3358216}

\bibitem{sage1998systems}
Sage, A.P., Lynch, C.L.: Systems integration and architecting: An overview of
  principles, practices, and perspectives. Systems Engineering: The Journal of
  The International Council on Systems Engineering  \textbf{1}(3),  176--227
  (1998)

\bibitem{sangiovanni2012taming}
Sangiovanni-Vincentelli, A., Damm, W., Passerone, R.: Taming dr. frankenstein:
  Contract-based design for cyber-physical systems. European journal of control
   \textbf{18}(3),  217--238 (2012)

\bibitem{SANGIOVANNIVINCENTELLI2012217}
Sangiovanni-Vincentelli, A.L., Damm, W., Passerone, R.: Taming {D}r.
  {F}rankenstein: {C}ontract-{B}ased {D}esign for {C}yber-{P}hysical {S}ystems.
  European Journal of Control  \textbf{18}(3),  217 -- 238 (2012).
  \doi{http://dx.doi.org/10.3166/ejc.18.217-238},
  \url{http://www.sciencedirect.com/science/article/pii/S0947358012709433}

\end{thebibliography}

\appendix

\section{Proofs}
\begin{proof}[Proof of Proposition \ref{bkhlx}]
    Let $\cont_i$ be the contract stated in the proposition.
    Observe that $\cont_i \land \cont' = \cont \land \cont'$. Thus, by the monotonicity of conjunction, $\cont'' \le \cont_i$ implies that $\cont'' \land \cont' \le \cont$.

    Now write $\cont''$ as $\cont'' = (A'', G'')$ and assume that $\cont'' \land \cont' \le \cont$. Then
    \begin{align}
        \nonumber
        &G'' \setint G' \le G \text{ and} \\
        \nonumber
        &A'' \setunion A' \ge A.
        \intertext{From this we conclude that}
        \label{{\documentPath}lkjxnahkl}&G'' \le G \setunion \neg G' \text{ and} \\
        \label{{\documentPath}kjaha}&A'' \ge A \setint \neg A'.
        \intertext{From~\eqref{{\documentPath}lkjxnahkl} and the fact that $A'' \ge \neg G''$ (which follows from the definition of AG contracts), we obtain $A'' \ge G' \setint \neg G$. This result and~\eqref{{\documentPath}kjaha} yield}
        \nonumber
        &A'' \ge (A \setint \neg A') \setunion (G' \setint \neg G).
    \end{align}
    This expression and~\eqref{{\documentPath}lkjxnahkl} mean that $\cont'' \le \cont_i$, completing the proof.
\end{proof}

\begin{proof}[Proof of Proposition \ref{kahfllnk}]
We already know that $1$ and $0$ are, respectively, the identity elements of conjunction and disjunction. $\id$ is the identity for composition and merging. The idempotence of these operations follows immediately from their definitions. It remains to show that these operations are associative.

Let $\cont = (a, g)$, $\cont' = (a', g')$, and $\cont'' = (a'', g'')$ be contracts.
\begin{itemize}
    \item Conjunction.
    \begin{align*}
        \cont & \land (\cont' \land \cont'') =
        (a, g) \land (a' \lor a'', g' \land g'') \\ &=
        \left( a \lor (a' \lor a''), g \land (g' \land g'') \right) \\ &=
        \left( (a \lor a') \lor a'', (g \land g') \land g'' \right) \\ &=
        \left( a \lor a', g \land g' \right) \land \cont'' = 
        (\cont \land \cont') \land \cont''
    \end{align*}
    \item Disjunction.
    \begin{align*}
        \cont &\lor (\cont' \lor \cont'') =
        \left( \cont^{-1} \land ((\cont')^{-1} \land (\cont'')^{-1}) \right)^{-1}  \\ &=
        \left( (\cont^{-1} \land (\cont')^{-1}) \land (\cont'')^{-1} \right)^{-1} =
        (\cont \lor \cont') \lor \cont''
    \end{align*}
    \item Composition.
    \begin{align*}
        \cont & \parallel (\cont' \parallel \cont'') \\ &=
        (a, g) \parallel (\neg g' \lor \neg g'' \lor (a' \land a''), g' \land g'') \\ &=
        (\neg g \lor \neg g' \lor \neg g'' \lor (a \land a' \land a''), g \land (g' \land g'')) \\ &=
        (\neg g \lor \neg g' \lor \neg g'' \lor ((a \land a') \land a''), (g \land g') \land g'') \\ &=
        (\neg g \lor \neg g' \lor (a \land a') , g \land g') \parallel \cont'' =
        (\cont \parallel \cont') \parallel \cont''
    \end{align*}
    \item Merging.
    \begin{align*}
        \cont & \bullet (\cont' \bullet \cont'') =
        \left( \cont^{-1} \parallel ((\cont')^{-1} \parallel (\cont'')^{-1}) \right)^{-1} \\ &=
        \left( (\cont^{-1} \parallel (\cont')^{-1}) \parallel (\cont'')^{-1} \right)^{-1} =
        (\cont \bullet \cont') \bullet \cont'' \qedhere
    \end{align*}
\end{itemize}
\end{proof}

\begin{proof}[Proof of Proposition \ref{nkjhlbskf}]
    Due to the duality relations between conjunction and disjunction and between composition and merging,
    the reciprocal map provides monoid isomorphisms between $(\calg(\balg), \land, 1)$ and $(\calg(\balg), \lor, 0)$ and between
    $(\calg(\balg), \parallel, \id)$ and $(\calg(\balg), \bullet, \id)$.

    We now show that the map $\theta_g: \parmon(\balg) \to \andmon(\balg)$ defined as
    $$\theta_g: (a, g) \mapsto (\neg(a\land g), g)$$
    is a monoid isomorphism.

    Observe that $\theta_g^2(a, g) = \theta_g (\neg(a\land g), g) = (\neg(\neg(a\land g) \land g), g) = (a, g)$, so $\theta_g$ is an involution. We proceed to show it is a monoid homomorphism. Let $\cont = (a, g)$ and $\cont' = (a', g')$.
    \begin{itemize}
        \item $\theta_g (\id) = \theta_g(1, 1) = (0, 1) = 1$
        \item We verify whether $\theta_g$ commutes with the multiplications:
        \begin{align*}
        \theta_g&(\cont \parallel \cont') = \theta_g(\neg(g \land g') \lor (a \land a'), g\land g') \\ &=
        (\neg(a \land g \land a'\land g'), g \land g') \\ &=
        (\neg(a \land g) \lor \neg (a'\land g'), g \land g') \\ &=
        (\neg(a \land g), g) \land (\neg(a' \land g'), g') =
        \theta_g(\cont) \land \theta_g(\cont').
        \end{align*}
    \end{itemize}
    As $\theta_g$ is an involution, we have to check that it is a monoid map from $\andmon(\balg)$ to $\parmon(\balg)$.
    \begin{align*}
        \theta_g&(\cont \land \cont') \\ &= \theta_g(a \lor a', g\land g') =
        (\neg(g \land g' \land (a\lor a')), g \land g') \\ &=
        (\neg(g \land g') \lor (\neg a\land \neg a'), g \land g') \\ &=
        ( \neg (g \land g') \lor \left(\neg(a \land g) \land \neg(a' \land g')\right) , g \land g') \\ &=
        ( \neg(a \land g), g) \parallel ( \neg(a' \land g'), g') =
        \theta_g(\cont) \parallel \theta_g(\cont')
    \end{align*}

    The isomorphism $\theta_a$ is defined using the diagram \eqref{lkjvsfv}, i.e.,
$$\theta_a (a, g) = \left( \theta_g (a, g)^{-1}\right)^{-1} = \left( \theta_g (g, a) \right)^{-1} =
(\neg(a \land g), a)^{-1} = (a, \neg(a \land g)).\qedhere$$
\end{proof}

\begin{proof}[Proof of Theorem \ref{bnwjdnhx}]
    Because $\iota$ is monic, $f$ generates a unique monoid map $f^\#$:
$$
\begin{tikzcd}
    \boolandmon(\balg) \times \boolandmon(\balg) \arrow[d, "\pi", two heads] & \boolandmon(\balg') \times \boolandmon(\balg') \\
    \parmon(\balg) \arrow[r, "f"'] \arrow[ur, "f^\#"', dashed]
    & \parmon(\balg') \arrow[u, "\iota"', rightarrowtail]
\end{tikzcd}
$$
Because $\pi$ is epic, $f^\#$ generates a unique monoid map $f^\flat$
$$
\begin{tikzcd}
    \boolandmon(\balg) \times \boolandmon(\balg) \arrow[d, "\pi", two heads] 
    \arrow[r, "f^\flat", dashed]
    & \boolandmon(\balg') \times \boolandmon(\balg') \\
    \parmon(\balg) \arrow[r, "f"'] \arrow[ur, "f^\#"', dashed]
    & \parmon(\balg') \arrow[u, "\iota"', rightarrowtail]
\end{tikzcd}
$$
Thus, we have the diagram
\begin{equation}\label{ifbwkrfskdbg}
    \begin{tikzcd}
        \boolandmon(\balg) \times \boolandmon(\balg)
        \arrow[d, "\pi", two heads, shift left=0.7ex]
        \arrow[r, "f^\flat", dashed]
        & \boolandmon(\balg') \times \boolandmon(\balg') 
        \arrow[d, "\pi"', two heads, shift right=0.7ex]    
        \\
        \parmon(\balg) \arrow[r, "f"'] \arrow[ur, "f^\#"', dashed]
        \arrow[u, "\iota", rightarrowtail, shift left=0.7ex]
        & \parmon(\balg') \arrow[u, "\iota"', rightarrowtail, shift right=0.7ex]
    \end{tikzcd}
\end{equation}

$f^\flat$ can be factored as the product of two maps $$\boolandmon(\balg) \times \boolandmon(\balg) \to \boolandmon(\balg').$$
We also observe that
$f^\flat(a, g) = f^\flat ((a, 1) \land (1, g)) = f^\flat (a, 1) \land f^\flat(1, g)$.
This means there are monoid maps $l_a, l_g, r_a, r_g\colon \boolandmon(\balg) \to \boolandmon(\balg')$ such that
\[
f^\flat (a, g) = \left( l_a(a) l_g(g) , r_a(a) r_g(g) \right).
\]
To obtain further restrictions on these maps, we use~\eqref{ifbwkrfskdbg}:
$
    f^\flat (a, g) = \iota \circ f \circ \pi (a, g) = 
    \iota \circ \pi \circ f^\flat \circ \iota \circ \pi (a, g) 
    = \left( l_a(ag) l_g(g) r_a(ag) r_g(g) , r_a(ag) r_g(g) \right)
$.
\end{proof}

\begin{proof}[Proof of Proposition \ref{qocvksf}]
Let $\cont = (a, g)$, $\cont' = (a', g')$, and $\cont'' = (a'', g'')$ be contracts.
\begin{itemize}[leftmargin=*]
    \item Conjunction.
    \begin{align*}
        \cont & \land (\cont' \lor \cont'') =
        (a, g) \land (a' \land a'', g' \lor g'') \\ &=
        \left((a \lor a') \land (a \lor a''), (g\land g') \lor (g \land g'')\right) \\ &=
        (\cont \land \cont') \lor (\cont \land \cont'')\\
        \cont & \land (\cont' \parallel \cont'') \\ &=
        (a, g) \land \left ((g' \land g'') \to (a' \land a''), g' \land g'' \right) \\ &=
        \left (a \lor (a' \land a'') \lor \neg g' \lor \neg g'', g \land g' \land g'' \right) 
        \\ &=
        \left (a \lor \neg g \lor (a' \land a'') \lor \neg g' \lor \neg g'', g \land g' \land g'' \right)
        \\ &=
        \left ((g \land g' \land g'') \to ((a \lor a') \land (a \lor a'')), g \land g' \land g'' \right) \\ &=
        (a \lor a', g \land g') \parallel (a \lor a'', g \land g'') \\ &=
        (\cont \land \cont') \parallel (\cont \land \cont'')\\
        \id & \land (1 \bullet 0) =
        \id \land 1 =
        \id \ne
        0 \\ &= \id \bullet 0 = (\id \land 1) \bullet (\id \land 0)
    \end{align*}
    \item Disjunction.
    \begin{align*}
        \cont &\lor (\cont' \land \cont'') =
        \left( \cont^{-1} \land ((\cont')^{-1} \lor (\cont'')^{-1}) \right)^{-1} \\ &=
        \left( (\cont^{-1} \land (\cont')^{-1}) \lor ( \cont^{-1} \land (\cont'')^{-1}) \right)^{-1} \\ &=
        (\cont \lor \cont') \land (\cont \lor \cont'') \\
        \cont &\lor (\cont' \bullet \cont'') =
        \left( \cont^{-1} \land ((\cont')^{-1} \parallel (\cont'')^{-1}) \right)^{-1} \\ &=
        \left( (\cont^{-1} \land (\cont')^{-1}) \parallel ( \cont^{-1} \land (\cont'')^{-1}) \right)^{-1} \\ &=
        (\cont \lor \cont') \bullet (\cont \lor \cont'') \\
        \id &\lor (1 \parallel 0) =
        \id \lor 0 =
        \id \ne
        1 \\ &= 1 \parallel \id = (\id \lor 1) \parallel (\id \lor 0)
    \end{align*}
    \item Composition.
    \begin{align*}
        \cont \parallel & (\cont' \land \cont'') =
        (a, g) \parallel (a' \lor a'', g' \land g'') \\ =&
        \left(\begin{aligned} &\neg (g\land g') \lor \neg (g \land g'') \lor (a \land a') \lor (a \land a''), \\ &(g\land g') \land (g \land g'')\end{aligned}\right) \\ =&
        \left((g\land g') \to (a \land a') , (g\land g') \right) \land \\ &
        \left((g \land g'') \to (a \land a''), (g \land g'')\right) \\
        =&
        (\cont \parallel \cont') \land (\cont \parallel \cont'')\\
        \cont  \parallel & (\cont' \lor \cont'') =
        (a, g) \parallel (a' \land a'', g' \lor g'') \\ =&
        \left( \begin{aligned} & \neg (g \land g') \land \neg (g \land g'') 
        \lor ((a \land a') \land (a \land a''))
        , \\ & (g \land g') \lor (g \land g'') \end{aligned} \right) \\ =&
        \left( \begin{aligned} & \left( \neg (g \land g') \lor (a \land a') \right) \land
        \left( \neg (g \land g'') \lor (a \land a'') \right)
        , \\ & (g \land g') \lor (g \land g'') \end{aligned} \right) \\
        =&
        \left((g\land g') \to (a \land a') , (g\land g') \right) \\ & \lor
        \left((g \land g'') \to (a \land a''), (g \land g'')\right) \\
        =&
        (\cont \parallel \cont') \lor (\cont \parallel \cont'')\\
        1 \parallel & (0 \bullet \id) =
        1 \parallel 0 = 0 \ne 1 =
        0 \bullet 1 \\ =& (1 \parallel 0) \bullet (1 \parallel \id)
    \end{align*}
    The distributivity of composition over conjunction was shown in~\cite{10.1007/978-3-030-89247-0_3}.
    \item Merging.
    \begin{align*}
        \cont &\bullet (\cont' \land \cont'') =
        \left( \cont^{-1} \parallel ((\cont')^{-1} \lor (\cont'')^{-1}) \right)^{-1} \\ &=
        \left( (\cont^{-1} \parallel (\cont')^{-1}) \lor ( \cont^{-1} \parallel (\cont'')^{-1}) \right)^{-1} \\ &=
        (\cont \bullet \cont') \land (\cont \bullet \cont'') \\
        \cont &\bullet (\cont' \lor \cont'') =
        \left( \cont^{-1} \parallel ((\cont')^{-1} \land (\cont'')^{-1}) \right)^{-1} \\ &=
        \left( (\cont^{-1} \parallel (\cont')^{-1}) \land ( \cont^{-1} \parallel (\cont'')^{-1}) \right)^{-1} \\ &=
        (\cont \bullet \cont') \lor (\cont \bullet \cont'') \\
        0 &\bullet (1 \parallel \id) =
        0 \bullet 1 =
        1 \ne
        0 = 1 \parallel 0 \\ &= (0 \bullet 1) \parallel (0 \bullet \id) \qedhere
    \end{align*}
\end{itemize}
\end{proof}

\begin{proof}[Proof of Proposition \ref{nkkbhgfb}]
    Due to the distributivity of $\star$ over conjunction and disjunction (Proposition~\ref{qocvksf}), we have
    $
    (C_1 \land C_2) \star (C_3 \land C_4) = 
    (C_1 \star C_3) \land (C_1 \star C_4) \land 
    (C_2 \star C_3) \land (C_2 \star C_4) \le
    (C_1 \star C_3) \land (C_2 \star C_4)
    $ and
    $
    (C_1 \lor C_2) \star (C_3 \lor C_4) = 
    (C_1 \star C_3) \lor (C_1 \star C_4) \lor 
    (C_2 \star C_3) \lor (C_2 \star C_4) \ge
    (C_1 \star C_3) \lor (C_2 \star C_4)
    $.
\end{proof}

\begin{proof}[Proof of Proposition \ref{vfglgqbwf}]
    Tables~\ref{{\documentPath}nyfclfr} and~\ref{{\documentPath}kcsdy} tell, respectively, how operations behave with respect to the distinguished elements and how operations distribute.

    Suppose conjunction is the multiplication operation. Since $\cont \land \id \ne \id$, neither merging nor composition can be the addition operations. On the other hand, $\cont \land 0 = 0$, and conjunction distributes over disjunction. Thus, $(\calg(\balg), \land, \lor, 1, 0)$ is a semiring.

    Now we assume disjunction is the multiplication operation. Since $\cont \lor \id \ne \id$, neither merging nor composition can be the addition operations. However, $\cont \lor 1 = 1$, and disjunction distributes over conjunction. Thus, $(\calg(\balg), \lor, \land, 0, 1)$ is a semiring.

    Suppose composition is the multiplication operation. Since composition does not distribute over merging, merging cannot be addition. Since $\cont \parallel 1 \ne 1$, conjunction cannot be addition. However, $\cont \parallel 0 = 0$ and composition distributes over disjunction. Thus, $(\calg(\balg), \parallel, \lor, \id, 0)$ is a semiring.

    Now suppose that merging is the multiplication. Since merging does not distribute over composition, composition cannot be addition. Also, since $\cont \bullet 0 \ne 0$, conjunction cannot be addition. However, $\cont \bullet 1 = 1$ and merging distributes over conjunction. Thus, $(\calg(\balg), \bullet, \land, \id, 1)$ is a semiring.
\end{proof}

\begin{proof}[Proof of Proposition \ref{nhwlfchcjbg}]
    The two semiring isomorphisms are given by the reciprocal map.
    Suppose there is a semiring map $\beta\colon \parsr(\balg) \to \andsr(\balg)$.
    Then
    \begin{align*}\beta(a, 1_B) = \beta((a, 1_B) \lor \id) = \beta(a, g) \lor 1 = 1,\end{align*}
    which means that $\beta$ is not invertible.
\end{proof}

\begin{proof}[Proof of Proposition \ref{khjbcgfkbgfk}]
    Let $b, b' \in \balg$.
    \begin{itemize}
        \item $\Delta_g (0_B) = (1_B, 0_B)  = 0$ and $\Delta_g (1_B) = (0_B, 1_B)  = 1$
        \item $\Delta_g (b \land b') = (\neg(b \land b'), b \land b') = (\neg b \lor \neg b', b \land b') = \Delta_g(b) \land \Delta_g(b')$ 
        \item $\Delta_g (b \lor b') = (\neg(b \lor b'), b \lor b') = (\neg b \land \neg b', b \lor b') = \Delta_g(b) \lor \Delta_g(b')$
    \end{itemize}
    This shows that $\Delta_g$ is a semiring homomorphism. Now we study $\iota_g$:
    \begin{itemize}
        \item $\iota_g (0_B) = (1_B, 0_B)  = 0$ and $\iota_g (1_B) = (1_B, 1_B)  = \id$
        \item $\iota_g (b \land b') = (1_B, b \land b') = (1_B, b) \parallel (1_B, b') = \iota_g(b) \parallel \iota_g(b')$ 
        \item $\iota_g (b \lor b') = (1, b \lor b') = (1_B, b) \lor (1_B, b')= \iota_g(b) \lor \iota_g(b')$
    \end{itemize}
    We conclude that $\iota_g$ is a semiring homomorphism as well.
\end{proof}

\begin{proof}[Proof of Proposition \ref{knxfghsqblr}]
Let $b, b' \in \balg$, $\cont = (a, g)$, and $\cont' = (a', g')$. We have the following properties:
\begin{itemize}[leftmargin=*]
    \item Order.
    \begin{align*}
        b \cdot \cont &= b \cdot (a, g) = (b\land a, b \to g) \ge (a, g) = \cont \\
        \cont \cdot b &= (a, g) \cdot  b = (b\to a, b \land g) \le (a, g) = \cont
    \end{align*}
    Now suppose $\cont = (a, g) \le \cont' = (a', g')$. We have $a' \le a$ and $g \le g'$.
    Since the operations $b \land (\cdot)$ and $b \rightarrow (\cdot)$ are monotonic, we have
    $b \cdot \cont \le b \cdot \cont'$ and $\cont \cdot b \le \cont' \cdot b$.
    \item Reciprocal.
    \begin{align*}
        (b \cdot \cont)^{-1} &= (b \land a, b \to g)^{-1} = (b \to g, b \land a) \\ &= 
        (g, a) \cdot b = \cont^{-1} \cdot b
    \end{align*}
    \item Associativity.
    \begin{align*}
        (b & \land b') \cdot \cont =
        \left( (b \land b') \land a, (b \land b') \to g \right) \\ & =
        \left( b \land (b' \land a), b \to (b' \to g) \right) \\ &=
        b \cdot \left(b' \cdot (a,g) \right) =
        b \cdot \left(b' \cdot \cont \right)\\
        \cont & \cdot (b \land b') =
        \left((b \land b') \cdot \cont^{-1}\right)^{-1} =
        \left((b' \land b) \cdot \cont^{-1}\right)^{-1} \\ &=
        \left(b' \cdot \left(b \cdot \cont^{-1} \right)\right)^{-1} =
        \left(b \cdot \cont^{-1} \right)^{-1} \cdot b' =
        \left(\cont \cdot b  \right) \cdot b'
    \end{align*}
    \item Distributivity over the Boolean algebra $\balg$.
    \begin{align*}
        (b \lor b') \cdot \cont &=
        \left( (b \lor b') \land a, (b \lor b') \to g \right) \\ &=
        \left( (b \land a) \lor (b' \land a), (b \to g) \land (b' \to g) \right)  \\ &=
        (b\cdot \cont) \land (b' \cdot \cont) \\
        \cont \cdot (b \lor b') &=
        \left( (b \lor b') \cdot \cont^{-1} \right)^{-1} \\ &=
        \left( b \cdot \cont^{-1} \land b' \cdot \cont^{-1} \right)^{-1} \\ &=
        \cont \cdot b \lor \cont \cdot b'
    \end{align*}
    \item Distributivity over the contract operations.
    \begin{itemize}
        \item Conjunction.
            \begin{align*}
            b \cdot (\cont \land \cont') &=
            (b \land (a \lor a'), b\to (g \land g')) \\ &=
            ((b \land a) \lor (b \land a'), (b\to g) \land (b \to g')) \\ &=
            b\cdot \cont \land b\cdot \cont' \\
            (\cont \land \cont') \cdot b &=
            (b \to (a \lor a'), b \land (g \land g')) \\ &=
            ((b \to a) \lor a', (b \land g) \land g') \\ &=
            (\cont \cdot b) \land \cont'
            \end{align*}
        \item Disjunction.
            \begin{align*}
            b \cdot (\cont \lor \cont') &=
            \left( \left( \cont^{-1} \land (\cont')^{-1} \right) \cdot b\right)^{-1} \\ &=
            \left( \cont^{-1} \cdot b \land (\cont')^{-1}  \right)^{-1} \\ &=
            b \cdot \cont \lor \cont' \\
            (\cont \lor \cont') \cdot b &=
            \left( b \cdot \left( \cont^{-1} \land (\cont')^{-1} \right) \right)^{-1} \\ &=
            \left( b \cdot \cont^{-1} \land b \cdot (\cont')^{-1}  \right)^{-1} \\ &=
            \cont \cdot b \lor \cont' \cdot b
            \end{align*}
        \item Composition.
        \begin{align*}
            b &\cdot (\cont \parallel \cont') \\ &=
            (b \land \left( (g \land g') \to (a \land a') \right), b\to (g \land g')) \\ &=
            \left(\begin{aligned} & \left( b \to (g \land g') \right) \to (b \land a \land a') , \\ &b\to (g \land g') \end{aligned}\right) \\ &=
            \left(\begin{aligned} & \left( (b \to g) \land (b \to g') \right) \to ((b \land a) \land (b \land a')) , \\ & (b\to g) \land (b \to g') \end{aligned} \right) \\ &=
            b \cdot \cont \parallel b\cdot \cont'
            \\
            (\cont & \parallel \cont') \cdot b \\ &=
            \left( b \to \left( (g\land g') \to (a \land a') \right), b \land g \land g' \right) \\ &=
            \left( (b \land g\land g') \to \left( a \land a' \right), b \land g \land g' \right) \\ &=
            \left( (b \land g\land g') \to \left( (b \to a) \land a' \right), b \land g \land g' \right) \\ &=
            (\cont \cdot b) \parallel \cont'
            \end{align*}
        \item Merging.
            \begin{align*}
                b & \cdot (\cont \bullet \cont') =
                \left( (\cont^{-1} \parallel (\cont')^{-1}) \cdot b \right)^{-1} \\ &=
                \left( \cont^{-1} \cdot b \parallel (\cont')^{-1} \right)^{-1} =
                b \cdot \cont \bullet \cont' \\
                (\cont & \bullet \cont') \cdot b=
                \left( b \cdot (\cont^{-1} \parallel (\cont')^{-1}) \right)^{-1} \\ &=
                \left( b \cdot \cont^{-1}  \parallel b \cdot (\cont')^{-1} \right)^{-1} =
                \cont \cdot b \bullet \cont' \cdot b
            \end{align*}
    \end{itemize}
    \item Distributivity over the adjoint operations.
    \begin{itemize}
        \item Quotient.
            \begin{align*}
            b \cdot (\cont / \cont') &= b \cdot (\cont \bullet (\cont')^{-1}) =
            b \cdot \cont \bullet (\cont' )^{-1} \\ &=
            (b \cdot \cont) / (\cont' ) \\
            b \cdot (\cont / \cont') &= b \cdot (\cont \bullet (\cont')^{-1}) =
            \cont \bullet b \cdot (\cont')^{-1} \\ & =
            \cont \bullet (\cont' \cdot b)^{-1} =
            \cont / (\cont' \cdot b) \\
            (\cont / \cont') \cdot b &=
            (\cont \bullet (\cont')^{-1}) \cdot b =
            \cont \cdot b \bullet (\cont')^{-1} \cdot b \\ &=
            \cont \cdot b \bullet (b \cdot \cont')^{-1} =
            (\cont \cdot b) / (b \cdot \cont')
            \end{align*}
        \item Separation.
        \begin{align*}
            b \cdot (\cont \div \cont') &= b \cdot (\cont \parallel (\cont')^{-1}) =
            b \cdot \cont \parallel b \cdot (\cont')^{-1} \\ &=
            b \cdot \cont \parallel (\cont' \cdot b)^{-1} =
            (b \cdot \cont) \div (\cont' \cdot b) \\
            (\cont \div \cont') \cdot b &=
            (\cont \parallel (\cont')^{-1}) \cdot b  =
            \cont \cdot b \parallel (\cont')^{-1}  \\ &=
            (\cont \cdot b) \div \cont' \\
            (\cont \div \cont') \cdot b &=
            (\cont \parallel (\cont')^{-1}) \cdot b
            \\ &=
            \cont \parallel (\cont')^{-1} \cdot b =
            \cont  \div (b \cdot \cont')
            \end{align*}
        \item Implication.
        \begin{align*}
            b & \cdot (\cont' \to \cont) =
            b \cdot \left( (a\land \neg a') \lor (g' \land \neg g), g \lor \neg g'\right) \\ &=
            \left( b \land (a\land \neg a') \lor (b \land g' \land \neg g), (b \to g) \lor \neg g'\right) \\ &=
            \left( \begin{aligned} & ((b \land a)\land \neg a') \lor (g' \land \neg (b \to g)), \\ & (b \to g) \lor \neg g' \end{aligned}\right) \\ &=
            \cont' \to b \cdot \cont \\
            b &\cdot (\cont' \to \cont) =
            b \cdot \left( (a\land \neg a') \lor (g' \land \neg g), g \lor \neg g'\right) \\ &=
            \left(\begin{aligned} & (a\land \neg (b \to a')) \lor (b \land g' \land \neg g), \\ & (b\land g') \to g \end{aligned}\right) \\ &=
            \cont' \cdot b \to \cont \\
            (\cont' & \to \cont) \cdot b=
            \left( (a\land \neg a') \lor (g' \land \neg g), g \lor \neg g'\right) \cdot b \\ &=
            \left( \begin{aligned} & (b \to (a\land \neg a')) \lor (b \to (g' \land \neg g)), \\ & (b \land g) \lor \neg (b \to g') \end{aligned} \right) \\ &=
            \left( \begin{aligned} & (b \to a) \land \neg (b \land a') \lor (b \to g') \land \neg (b \land g), \\ & (b \land g) \lor \neg (b \to g') \end{aligned} \right) \\ &=
            b\cdot \cont' \to \cont \cdot b
        \end{align*}
        \item Coimplication.
        \begin{align*}
            b & \cdot (\cont' \nrightarrow \cont) = 
            \left( ((\cont')^{-1} \to \cont^{-1}) \cdot b\right)^{-1} \\ &=
            \left( b \cdot (\cont')^{-1} \to \cont^{-1} \cdot b \right)^{-1} =
            \cont' \cdot b \nrightarrow b \cdot \cont \\
            (\cont' & \nrightarrow \cont) \cdot b= 
            \left( b \cdot ((\cont')^{-1} \to \cont^{-1}) \right)^{-1} \\ &=
            \left( (\cont')^{-1} \to b \cdot \cont^{-1} \right)^{-1} =
            \cont'  \nrightarrow \cont \cdot b \\ &=
            \left( (\cont')^{-1} \cdot b \to \cont^{-1} \right)^{-1} =
            b \cdot \cont' \nrightarrow \cont \qedhere
        \end{align*}
    \end{itemize}
\end{itemize}
\end{proof}

\end{document}